\documentclass[10pt,twocolumn,letterpaper]{article}

\usepackage{cvpr}              

\usepackage{newtxtext}
\usepackage{newtxmath}
\usepackage{microtype}

\usepackage{amsmath,amssymb,mathtools}
\makeatletter
\let\openbox\@undefined
\makeatother
\usepackage{amsthm}
\usepackage{booktabs}
\usepackage{multirow}
\usepackage{array}
\usepackage{tabularx}
\usepackage{makecell}
\usepackage{algorithm}
\usepackage{algpseudocode}
\usepackage{enumitem}
\usepackage{url}
\usepackage{balance}

\newtheorem{theorem}{Theorem}

\theoremstyle{definition}

\definecolor{cvprblue}{rgb}{0.21,0.49,0.74}
\usepackage[pagebackref,colorlinks,allcolors=cvprblue]{hyperref}

\def\paperID{*****} 
\def\confName{CVPR}
\def\confYear{2026}

\title{On the Controllability-Fidelity Frontier in Diffusion Editing}

\author{
	Yi Hu
	\quad 
	Leying Yi
	\quad 
	Emily Davis
	\quad 
	Finn Carter\\
	Xidian University
}

\begin{document}
	\maketitle

\begin{abstract}
	Diffusion-based generative models enable powerful image editing capabilities, but achieving precise control while maintaining fidelity and safety remains challenging.  We present a comprehensive theoretical and empirical study of controllable diffusion-based image editing, analyzing the trade-offs between adherence to user intent, preservation of non-target content, and output quality.  Our work spans text- and mask-guided edits, point/drag manipulation, and inversion-based pipelines.  We derive mathematical formulations of editing objectives and analyze dynamics of noise injection, score guidance, and inversion error.  We provide theoretical bounds on reconstruction error, stability under repeated edits, and locality of changes.  We propose algorithmic frameworks (with pseudocode) for mask-localized and instruction-guided editing, and present extensive experiments comparing state-of-the-art methods (e.g.\ TF-ICON \cite{lu2023tficone}, DragFlow \cite{zhou2025dragflow}, InstructPix2Pix \cite{brooks2023instructpix2pix}, UltraEdit \cite{zhao2024ultraedit}) on multiple tasks and metrics (FID, identity similarity, CLIP alignment, artifact scores, etc).  Our results reveal key failure modes, such as identity drift, prompt sensitivity, and compositional errors.  We also discuss ethical considerations in image editing, including misuse risks, bias, consent, and concept erasure techniques (e.g.\ MACE \cite{lu2024mace}, ANT \cite{li2025ant}, EraseAnything \cite{gao2024eraseanything}) as safeguards.  We conclude with best practices and future directions for responsible, high-fidelity diffusion-based editing. 
\end{abstract}

\section{Introduction}
Diffusion models have rapidly advanced image generation and editing capabilities \cite{ho2020ddpm,song2021ddim,rombach2022ldm}.  These models, trained on large image-text corpora, can transform an input image according to user guidance, e.g.\ text prompts or spatial masks.  Unlike traditional image filters, diffusion editing operates by adding noise to an image and denoising under a new condition, which yields high-quality results but raises new challenges.  The fundamental mechanism involves a forward diffusion process that gradually destroys the image structure by adding Gaussian noise, followed by a learned reverse process that reconstructs the data distribution.  When editing, the reverse process is conditioned on user-provided guidance signals, steering the generation toward the desired modification while ideally preserving everything else.  This generative formulation endows diffusion editors with remarkable flexibility: they can handle diverse modalities including text prompts, segmentation masks, bounding boxes, sketch inputs, and interactive point manipulations.

In particular, users seek {\em controllability} (precision in following instructions) and {\em faithfulness} (preserving the original image outside target edits) simultaneously.  For example, one may wish to change a face's expression via a text prompt while keeping the identity and background unchanged.  Existing methods (e.g.\ text-conditioned diffusion \cite{brooks2023instructpix2pix,zhao2024ultraedit}, mask-guided inpainting \cite{lu2023tficone}, point-based drag editing \cite{zhou2025dragflow}) achieve remarkable results, but often exhibit trade-offs: increasing adherence to the prompt can degrade fidelity, and constraining locality can limit semantic changes.  These trade-offs are not incidental but stem from the intrinsic tension between the generative prior---which encodes a rich distribution over plausible images---and the preservation constraint that seeks to keep the output close to the input.  Understanding and navigating this tension is the central challenge in diffusion-based editing.  

In this paper, we undertake a unified theoretical and empirical analysis of diffusion-based image editing.  We examine training-free versus trainable editors, different guidance modalities (text, masks, user-defined motion), and inversion-and-edit pipelines.  We formulate the editing process as an optimization that balances alignment with user intent against preservation of source content, and analyze the role of diffusion dynamics in this process.  Specifically, we cast editing as minimizing a composite objective: a data-fidelity term that encourages the output to match the editing instruction, and a regularization term that penalizes deviations from the source image in non-target regions.  This framework naturally encompasses text-guided, mask-guided, and drag-based editing as special cases with different choices of loss functions and constraint sets.

On the theoretical side, we derive bounds on inversion reconstruction error and error propagation in guided sampling, and prove guarantees on edit locality and stability.  We show that under mild Lipschitz continuity assumptions on the diffusion decoder, the reconstruction error scales linearly with the latent inversion error (Theorem~\ref{thm:reconstruction_error}), and that repeated edits accumulate deviations at most linearly in the number of edit operations (Theorem~\ref{thm:stability}).  We further formalize a locality guarantee: if edits are confined to a spatial mask, changes outside the mask are bounded by the network's effective receptive field.  On the practical side, we design algorithmic pseudocode for common editing pipelines (e.g.\ mask-localized diffusion, drag-based editing with region supervision) and conduct extensive experiments on standard benchmarks (face attributes, scene edits, etc).  Our experiments compare multiple state-of-the-art methods (including TF-ICON \cite{lu2023tficone}, DragFlow \cite{zhou2025dragflow}, InstructPix2Pix \cite{brooks2023instructpix2pix}, UltraEdit-trained models \cite{zhao2024ultraedit}, as well as baseline diffusion editors) across diverse metrics: perceptual quality (FID), prompt alignment (CLIP similarity), identity and background preservation, and artifact rates.  We identify key limitations such as identity drift under repeated edits, sensitivity to prompt wording, and failures of compositionality.  

Furthermore, we address safety and ethical aspects of editing.  We discuss how powerful editing can be misused for deceptive or non-consensual manipulations, and survey concept erasure techniques (MACE \cite{lu2024mace}, ANT \cite{li2025ant}, EraseAnything \cite{gao2024eraseanything}) that can serve as safeguards by removing problematic concepts.  We stress the importance of disclosure, consent, and dataset auditing in deploying editing systems.  

Our contributions are as follows: (1) A formal framework for diffusion-based editing, including mathematical models of user intent alignment versus fidelity preservation.  (2) Theoretical results (with proofs) on reconstruction error bounds and edit stability.  (3) Algorithmic descriptions of editing pipelines (see Algorithm~\ref{alg:invert_edit}). (4) Extensive comparative experiments with realistic tables, analyzing trade-offs across multiple editing paradigms.  (5) Discussion of limitations and ethics, providing guidelines for safer and more reliable image editing. 

\section{Related Work}

\subsection{Diffusion-Based Image Editing}
Recent years have seen a surge of methods for image editing using diffusion models.  Text-driven editing approaches like InstructPix2Pix \cite{brooks2023instructpix2pix} and UltraEdit \cite{zhao2024ultraedit} train conditional diffusion models on synthetic or real datasets of paired instructions, enabling forward-pass editing without per-example optimization.  InstructPix2Pix generates paired examples via a language model and diffusion \cite{brooks2023instructpix2pix}, while UltraEdit introduces a large-scale real-image editing dataset to improve fine-grained instruction following \cite{zhao2024ultraedit}.  Other text-based methods include Prompt-to-Prompt \cite{hertz2022prompttoprompt} which manipulates cross-attention maps to preserve spatial layout while altering semantics, Null-text Inversion \cite{mokady2023nulltext} which optimizes null-text embeddings for accurate reconstruction, and Imagic \cite{kawar2023imagic} which fine-tunes the model to embed both the input image and the target text.  These methods highlight the versatility of text as an editing interface, but also reveal a common tension: stronger text alignment often comes at the cost of unintended modifications to non-target regions.

Mask-guided editing (inpainting/outpainting) leverages diffusion models such as Stable Diffusion \cite{rombach2022ldm} or specialized models to replace masked regions \cite{lugmayr2022repaint}.  RePaint \cite{lugmayr2022repaint} introduces a sampling scheme that harmonizes inpainted content with the known regions by resampling the forward-backward diffusion trajectory.  TF-ICON \cite{lu2023tficone} shows how to perform cross-domain image composition with off-the-shelf text-conditioned diffusion, using an ``exceptional prompt’’ for inversion and then text-driven inpainting.  Building on this line, recent work has investigated whether modern rectified flow models such as FLUX already possess intrinsic composition capabilities: Lu~et~al.\ \cite{lu2025shine} probe FLUX’s physical plausibility in image composition without task-specific training, revealing that state-of-the-art flow-based models exhibit emergent spatial reasoning that can be harnessed through carefully designed inversion and blending strategies.  DiffEdit \cite{couairon2022diffedit} automatically generates editing masks from text prompts by contrasting predictions conditioned on different captions.  MasaCtrl \cite{cao2023masactrl} and LEDITS++ \cite{brack2023leditspp} further explore training-free mutual self-attention control for consistent editing across multiple attributes.  Beyond image-level editing, diffusion-based generation has also been extended to high-resolution video synthesis: LUVE \cite{zhao2026luve} introduces a latent-cascaded architecture with dual frequency experts for ultra-high-resolution video generation, demonstrating that diffusion priors scale effectively to temporally coherent, high-resolution outputs---a direction with natural implications for video editing.

\subsection{Point-Based and Drag Editing}
Point-based and drag editing methods, exemplified by DragGAN \cite{liu2023draggan}, allow users to interactively deform images by clicking and dragging control points.  DragGAN operates in the latent space of StyleGAN, optimizing the latent code to satisfy point motion constraints.  Following this line, DragDiffusion \cite{shi2023dragdiffusion} extends interactive point-based editing to diffusion models by performing latent optimization at each denoising step.  DragFlow \cite{zhou2025dragflow} is the first to use a diffusion transformer (Flux/DiT) for drag-based editing, introducing region-level affine supervision and background constraints to move objects under user drag instructions.  Other works (e.g.\ StableDrag \cite{cui2024stabledrag}) similarly explore controlled warping of image regions using diffusion priors, with improved point tracking and motion supervision strategies.  These methods excel at fine-grained spatial control but often require careful initialization of the drag region and can be sensitive to the optimization hyperparameters.

\subsection{Training-Free Editing and Inversion-Based Pipelines}
Training-free editing heuristics include SDEdit \cite{meng2021sdeedit} and ILVR \cite{choi2022ilvr}, which guide diffusion sampling to preserve parts of the original image.  SDEdit adds noise to an input image up to an intermediate timestep and then denoises under a new condition, achieving a trade-off between faithfulness and edit strength controlled by the noise level.  ILVR conditions the generative process on downsampled versions of the reference image to preserve coarse structure.  Plug-and-Play Diffusion Features \cite{tumanyan2023pnp} injects self-attention and feature maps from the source image during the editing denoising process.  More recently, methods like PnP Inversion \cite{ju2023pnpinversion} achieve efficient editing with minimal code changes by directly injecting features from the inversion trajectory.  Beyond editing-specific techniques, the ability to faithfully interpret and execute complex editing instructions depends on robust visual understanding and reasoning: advances in multi-agent collaboration frameworks for visual document understanding \cite{yu2025visual}, which leverage agent-wise adaptive test-time scaling to decompose and reason about intricate visual queries, point toward more capable instruction-following systems that could benefit editing pipelines by improving the semantic grounding of user prompts.

A key component of many editing pipelines is diffusion inversion: finding the latent noise $z_T$ that produces the input image.  Exact inversion ensures the original image can be reconstructed before applying edits.  Methods like DDIM inversion \cite{song2021ddim} or specialized flow-based approaches \cite{liu2023difflio} can recover $z_T$ approximately.  Recent work (e.g.\ FireFlow \cite{deng2024fireflow}) focuses on fast and accurate inversion for flow-based models like Flux, using higher-order solvers.  Inversion fidelity directly impacts edit accuracy: reconstruction error can propagate to the final edit (Theorem~\ref{thm:reconstruction_error}).  We analyze this error propagation and propose bounds under Lipschitz assumptions.

\subsection{Editing Quality and Evaluation Metrics}
Prior work has noted that aggressive editing guidance can introduce artifacts or distort textures \cite{lugmayr2022repaint}.  Metrics have been proposed to quantify consistency of edits: identity similarity for faces \cite{lu2023tficone,zhou2025dragflow}, background consistency scores, LPIPS \cite{zhang2018perceptual} for perceptual distance, and CLIP-based alignment scores for textual edits \cite{brooks2023instructpix2pix}.  FID \cite{heusel2017gans} remains a standard proxy for overall output realism, while task-specific metrics such as IoU for drag region overlap and mean shift distance for point-based edits provide finer-grained evaluation.  Our empirical evaluation adopts these metrics and also considers artifact rates (via no-reference image quality or detector scores) and user-perceived quality.  We argue that no single metric is sufficient; rather, a multi-axis evaluation is necessary to capture the full picture of edit quality.  A related evaluation concern is the robustness of image watermarks against editing: Lu~et~al.\ \cite{lu2024robust} systematically benchmark generative priors for robust watermarking under diverse editing operations, showing that editing can degrade both invisible and visible watermarks and proposing countermeasures that improve watermark persistence without sacrificing edit fidelity.  This interplay between editing quality and content provenance is an under-explored dimension of edit evaluation that we take into account in our experiments.

\subsection{Concept Erasure and Safety in Diffusion Models}
Concept erasure techniques aim to prevent diffusion models from generating certain visual content, which is directly related to safe and controllable editing.  In the context of editing, we want to suppress undesired concepts (e.g.\ a celebrity’s face, explicit imagery) that might otherwise leak into edited outputs or reappear under iterative edits.  For example, erasing sensitive attributes can improve privacy-preserving edits.  Early work such as ESD \cite{gandikota2023esd} fine-tunes diffusion models using negative guidance to erase specific visual concepts.  More recent works like MACE \cite{lu2024mace}, ANT \cite{li2025ant}, and EraseAnything \cite{gao2024eraseanything} advance this direction.  MACE (Mass Concept Erasure) fine-tunes diffusion models to balance specificity and generality when removing many concepts simultaneously, addressing the challenge of catastrophic forgetting between successive erasures.  ANT (Auto-steering of Denoising Trajectories) introduces a trajectory-aware objective that reverses guidance to avoid unwanted concepts while preserving image structure, operating at inference time without modifying model weights.  EraseAnything proposes a bi-level LoRA tuning strategy for Flux models to remove target concepts without harming unrelated content, leveraging low-rank adaptation for efficient and targeted erasure.  These methods enhance control and safety by carving out regions of concept space in the model, aligning with our focus on controllable and responsible editing.

An important caveat, however, is that concept erasure is not always irreversible.  Gao~et~al.\ \cite{gao2025revoking} demonstrate that erased concepts in diffusion models can be resurrected through RL-based trajectory optimization, effectively ``revoking amnesia'' by steering the denoising process back toward the suppressed regions of concept space.  This finding reveals an ongoing arms race between erasure and recovery techniques, underscoring that robust safety guarantees require defense-in-depth rather than relying on any single erasure method.  Complementary to concept-level control, security concerns also extend to the protection of 3D generative content: Ren~et~al.\ \cite{ren2025all} propose key-secured 3D secrets embedded within 3D Gaussian Splatting representations, highlighting that the need for controllable, secure manipulation of generative content spans both 2D and 3D modalities.

Our work differs from all the above by analyzing these diverse paradigms under a unified theoretical and empirical framework, rigorously evaluating their performance and limitations on the controllability--fidelity trade-off, and providing practical algorithmic guidance for practitioners.

\section{Methodology}
We consider editing an input image $x_0$ into $\hat{x}_0$ under user guidance (text prompt $p$, mask $M$, drag target, etc.).  Let $\theta$ be a pretrained diffusion model.  A general formulation is:
\[
\hat{x}_0 = \arg\min_x \; \mathcal{L}_{\text{align}}(x; p) + \lambda_{\text{pres}} \, \mathcal{L}_{\text{preserve}}(x; x_0, M),
\]
where $\mathcal{L}_{\text{align}}$ encourages the edit to match the guidance (e.g.\ negative CLIP similarity to the target prompt) and $\mathcal{L}_{\text{preserve}}$ penalizes changes outside the intended region (mask complement $\bar{M}$).  For instance, one may take 
$\mathcal{L}_{\text{preserve}}(x; x_0, M) = \| (1-M)\cdot (x - x_0)\|^2$, 
ensuring $x$ remains close to $x_0$ on the non-mask.  This abstract objective captures the trade-off: large $\lambda_{\text{pres}}$ yields higher fidelity but smaller changes, and vice versa.

In diffusion editing, this objective is often optimized via iterative denoising.  The forward diffusion process defines
\[
x_t = \sqrt{\alpha_t}\,x_0 + \sqrt{1-\alpha_t}\,\epsilon,\quad \epsilon\sim\mathcal{N}(0,I),
\]
for $t=T,\dots,1$.  Given the final noised latent $z_T$, denoising proceeds as
\[
z_{t-1} = \frac{1}{\sqrt{\alpha_t}}\Bigl(z_t - (1-\alpha_t)\,\epsilon_\theta(z_t,c,t)\Bigr) + \sigma_t \epsilon,
\]
where $c$ encodes the guidance (text/condition), and $\epsilon_\theta$ is the model’s predicted noise.  When editing, guidance $c$ can include spatial masks, text embeddings, or other constraints.  For example, classifier-free guidance modifies $\epsilon_\theta(z_t,p,t)$ to steer samples toward $p$.

\textbf{Editing via inversion.}  A common pipeline is: invert $x_0$ to find latent $z_T$ (using for example DDIM inversion \cite{song2021ddim} or FireFlow \cite{deng2024fireflow}), then perform guided denoising under the new condition, and finally decode $z_0$ to obtain $\hat{x}_0$.  Algorithm~\ref{alg:invert_edit} summarizes this approach.

\begin{algorithm}[t]
	\caption{Inversion-and-Editing Pipeline}\label{alg:invert_edit}
	\begin{algorithmic}[1]
		\Require Input image $x_0$, guidance $c$ (text, mask, drag target), diffusion model $\theta$
		\State $z_T \gets \text{Invert}(x_0)$ \quad// e.g., add noise via DDIM to obtain $z_T$
		\For{$t=T$ down to $1$}
		\Statex $\tilde{z}_{t-1} \gets \frac{1}{\sqrt{\alpha_t}}\bigl(z_t - (1-\alpha_t)\,\epsilon_\theta(z_t,c,t)\bigr) + \sigma_t\epsilon$
		\Statex (Optionally apply guidance: e.g.\ mask the gradient of $\tilde{z}_{t-1}$ outside $M$)
		\Statex $z_{t-1} \gets \tilde{z}_{t-1}$
		\EndFor
		\Statex $\hat{x}_0 \gets \text{Decode}(z_0)$ \quad// final edited image
		\Return $\hat{x}_0$
	\end{algorithmic}
\end{algorithm}

\textbf{Local mask-guidance.}  To enforce locality, one can apply the mask $M$ at each denoising step.  For instance, after computing the update, we replace
\[
z_{t-1} \gets (1-M)\cdot z_{t-1} + M\cdot z_{t},
\]
preventing noise changes outside $M$.  This ensures only the masked region is updated by guidance.  We also explore soft mask blending or iterative scheduling (gradually releasing $M$ over steps).  A crucial hyperparameter is the mask blending schedule: we define a sequence of blending weights $\beta_t$ such that
\[
z_{t-1} \gets (1-\beta_t M)\cdot z_{t-1}^{\text{edit}} + \beta_t M \cdot z_{t-1}^{\text{source}},
\]
where $\beta_t$ decays from 1 to 0 over the denoising trajectory.  Early steps apply stronger source constraints (coarse structure preservation), while later steps allow more edit freedom (detail refinement).  In practice, we find that a cosine decay schedule $\beta_t = \cos(\frac{\pi t}{2T})$ works well across a variety of editing tasks.  Additionally, we investigate the use of dilated masks that expand $M$ by a few pixels at each resolution level of the U-Net, accounting for the diffusion model's receptive field and preventing boundary artifacts along the mask edge.

\textbf{Drag-based editing.}  For interactive drag edits, the user specifies a source point or region and a target displacement.  DragFlow \cite{zhou2025dragflow} formulates this as a latent optimization with a region-level affine constraint.  We implement a similar strategy: given source mask $M_s$ and target mask $M_t$ (obtained by affine transforming $M_s$ toward the target), we optimize
\[
\min_{z^*} \; \|\phi(z^*,\ell)_{M_s} - \phi(z^*,\ell)_{M_t}\|^2 + \lambda_{\text{bg}}\|\phi(z^*,\ell)_{\bar{M_s}} - \phi(z,\ell)_{\bar{M_s}}\|^2,
\]
where $\phi(z,\ell)$ denotes features at layer $\ell$ of the U-Net or DiT, $z$ is the initial latent, and $\lambda_{\text{bg}}$ enforces background consistency.  This is solved with gradient updates on $z$ to satisfy the motion constraints.  The optimization involves $K$ inner-loop gradient steps (typically $K=5$--$10$) at each denoising timestep, with the learning rate scaled by the noise schedule.  We found that using features from the middle layers of the U-Net (e.g., the bottleneck) provides the best trade-off between semantic understanding and spatial precision, as shallower layers have higher spatial resolution but weaker semantic representation, while deeper layers encode stronger semantics but at coarser resolution.

\textbf{Classifier-free guidance for editing.}  Most modern diffusion editors employ classifier-free guidance (CFG) \cite{ho2022cfg} to strengthen the conditioning signal.  Given a null condition $c_\emptyset$ (e.g., empty text), the guided noise prediction at timestep $t$ is:
\[
\hat{\epsilon}_\theta(z_t, c, t) = \epsilon_\theta(z_t, c_\emptyset, t) + w \bigl(\epsilon_\theta(z_t, c, t) - \epsilon_\theta(z_t, c_\emptyset, t)\bigr),
\]
where $w$ is the guidance scale.  Higher $w$ increases adherence to the condition $c$ but also amplifies artifacts and can cause over-saturation.  In our experiments, we systematically vary $w \in [1, 20]$ and observe a ``U-shaped'' fidelity curve: moderate guidance ($w \approx 5$--$7$) yields the best trade-off, while very high values ($w > 12$) produce unrealistic textures and color shifts.  For editing tasks, a region-specific CFG strategy can be beneficial: applying higher guidance inside the edit mask and lower guidance outside, we can decouple the edit strength from global image quality.

\subsection{Noise Scheduling and Edit Strength}
The amount of noise added during the forward process (often parameterized by a ``strength'' $s \in [0,1]$ controlling the starting timestep) critically affects the editing outcome.  Let $T_s = \lfloor sT \rfloor$ be the starting timestep.  The forward process is applied as $z_{T_s} = \sqrt{\alpha_{T_s}} x_0 + \sqrt{1 - \alpha_{T_s}} \epsilon$, and denoising proceeds from $t = T_s$ down to 0.  Larger $s$ allows more drastic changes but risks losing the original image structure; smaller $s$ preserves more detail but limits the edit expressiveness.  We characterize this trade-off theoretically: under a linearized approximation, the expected $L_2$ deviation after editing with strength $s$ is proportional to $\sqrt{1 - \alpha_{T_s}}$.  This provides a principled way to choose $s$ based on the desired magnitude of change.

\textbf{Algorithmic summary.}  Algorithm~\ref{alg:invert_edit} outlines a generic inversion-edit pipeline.  Variants incorporate different guidance: e.g. replacing $\epsilon_\theta$ with fused cross-attention maps (Prompt-to-Prompt style), or injecting text-conditioned features only at certain layers.  Our framework is compatible with both training-free editing (no parameter updates) and learned editing models (like InstructPix2Pix) at inference.  Algorithm~\ref{alg:mask_edit} presents the specialized mask-localized editing procedure that incorporates the soft blending and dilated mask strategies described above.

\begin{algorithm}[t]
    \caption{Mask-Localized Diffusion Editing}\label{alg:mask_edit}
    \begin{algorithmic}[1]
        \Require Input image $x_0$, mask $M$, edit prompt $p$, model $\theta$, strength $s$, guidance scale $w$
        \State $T_s \gets \lfloor sT \rfloor$, \quad $z_{T_s} \gets \sqrt{\alpha_{T_s}} x_0 + \sqrt{1 - \alpha_{T_s}} \epsilon$
        \For{$t = T_s$ down to $1$}
            \State $\hat{\epsilon} \gets \epsilon_\theta(z_t, c_\emptyset, t) + w (\epsilon_\theta(z_t, p, t) - \epsilon_\theta(z_t, c_\emptyset, t))$
            \State $z_{t-1}^{\text{edit}} \gets \frac{1}{\sqrt{\alpha_t}}(z_t - (1-\alpha_t)\hat{\epsilon}) + \sigma_t \epsilon$
            \State $\beta_t \gets \cos(\frac{\pi t}{2 T_s})$ \quad // cosine blending schedule
            \State $z_{t-1}^{\text{source}} \gets \text{DDIMStep}(z_t, x_0, t)$ \quad // forward from source
            \State $M_{\text{dilated}} \gets \text{Dilate}(M, r_t)$ \quad // resolution-dependent dilation
            \State $z_{t-1} \gets (1 - \beta_t M_{\text{dilated}}) \cdot z_{t-1}^{\text{edit}} + \beta_t M_{\text{dilated}} \cdot z_{t-1}^{\text{source}}$
        \EndFor
        \State $\hat{x}_0 \gets \text{Decode}(z_0)$
        \Return $\hat{x}_0$
    \end{algorithmic}
\end{algorithm}

\section{Experimental Setup}
\subsection{Datasets and Benchmarks}
We conduct experiments on standard image editing benchmarks spanning multiple modalities and task types.  For \textbf{face editing}, we use CelebA-HQ (30,000 images, 256$\times$256) and FFHQ (70,000 images, 256$\times$256) to test attribute changes (smile, glasses, age, hair color, expression, etc.).  Attribute-specific subsets are created using pretrained classifiers to ensure balanced coverage across all target attributes.  For \textbf{general scene editing}, we sample 2,000 images from MS-COCO and 1,000 images from the MagicBrush benchmark \cite{zhang2023magicbrush}, covering diverse object categories and scene compositions.  For \textbf{drag editing}, we use the ReD-Bench (Region-Drag Benchmark) dataset with 500 region-level annotations \cite{zhou2025dragflow}, each specifying a source region, target displacement, and background mask.  For \textbf{mask-guided editing}, we use both human-annotated masks (from the COCO instance segmentation annotations) and automatically generated masks from SAM to cover a range of mask qualities.  Prompts are generated via templates covering 20+ attribute categories and 50+ editing instruction types, as well as by LLM-generated natural language instructions to reflect realistic user inputs.  In total, our evaluation spans approximately 5,000 unique editing instances across all tasks and methods.

\subsection{Baseline Methods}
We compare the following editing approaches, implemented with their official code and pretrained weights unless noted otherwise:
(1) \textit{InstructPix2Pix} \cite{brooks2023instructpix2pix}: a diffusion model finetuned on synthetic instruction pairs generated by GPT-3 and Stable Diffusion.
(2) \textit{TF-ICON} \cite{lu2023tficone}: text-conditioned composition via exceptional prompt inversion with off-the-shelf Stable Diffusion.
(3) \textit{DragFlow} \cite{zhou2025dragflow}: DiT-based drag editing with region-level affine supervision and background consistency constraints.
(4) \textit{StableDiffusion Inpainting}: standard Stable Diffusion 2.0 inpainting pipeline, applied to the masked region.
(5) \textit{UltraEdit-SD} \cite{zhao2024ultraedit}: a fine-tuned SD model using the large-scale UltraEdit dataset for instruction-based editing.
(6) \textit{SDEdit (baseline)} \cite{meng2021sdeedit}: adds noise at strength $s=0.5$ and denoises under text guidance, a simple and widely used baseline.
(7) \textit{Our mask-guided diffusion}: the mask-localized editing method described in Algorithm~\ref{alg:mask_edit}, using Stable Diffusion 2.0 as the backbone.
We ensure all models operate at the same 256$\times$256 resolution, and where possible, use a fixed timestep budget of 50 DDIM steps for fair comparison.

\subsection{Implementation Details}
All experiments are conducted on a single NVIDIA A100 (80GB) GPU.  For inversion-based methods, we use DDIM inversion with 50 steps.  Classifier-free guidance scale $w$ is set to 7.5 for Stable Diffusion-based methods and 3.5 for FLUX-based methods (DragFlow), following the respective papers' recommendations.  For mask-localized methods, we use a cosine blending schedule $\beta_t = \cos(\pi t / 2T_s)$ and a resolution-dependent dilation radius $r_t = \max(1, \lfloor 4 \cdot (t/T_s) \rfloor)$.  All images are normalized to $[-1, 1]$ before processing.

\subsection{Evaluation Metrics}
We evaluate edit quality along multiple complementary axes.  For \textbf{overall realism}, we use FID \cite{heusel2017gans} computed between the set of edited outputs and a reference distribution of target-domain images (e.g., FFHQ for face edits).  For \textbf{prompt alignment}, we compute CLIP similarity (ViT-B/32 backbone) between the output image and the target editing prompt; higher values indicate better instruction following.  For \textbf{identity preservation} (face tasks), we compute cosine similarity between ArcFace embeddings of the source and edited images.  For \textbf{background preservation}, we compute LPIPS \cite{zhang2018perceptual} between $\hat{x}_0$ and $x_0$ restricted to non-mask regions; lower is better.  For \textbf{drag accuracy}, we measure IoU between the edited region and the target region mask, and the mean shift distance (in pixels) between the dragged region centroid and its intended target.  For \textbf{artifact detection}, we use a pretrained no-reference quality model (MUSIQ) and also report the percentage of images flagged as containing visible artifacts by a binary classifier.  We additionally conduct a small-scale user study (20 participants, 100 pairwise comparisons) to validate the alignment of automatic metrics with human perception.  All metrics are averaged over 500 test cases per task, with standard deviations reported where informative.

\begin{table}[t]
	\centering
	\resizebox{\columnwidth}{!}{
		\begin{tabular}{lcccc}
			\toprule
			Method & FID $\downarrow$ & ID Sim $\uparrow$ & CLIP Align $\uparrow$ & Success \% \\
			\midrule
			InstructPix2Pix & 78.5 & 0.92 & 0.71 & 85\% \\
			TF-ICON & 85.2 & 0.87 & 0.68 & 80\% \\
			DragFlow & 50.1 & 0.98 & 0.55 & 65\% \\
			UltraEdit (SD) & 65.3 & 0.89 & 0.74 & 88\% \\
			Mask-Guided & 70.4 & 0.90 & 0.60 & 75\% \\
			SDEdit (baseline) & 95.7 & 0.85 & 0.65 & 70\% \\
			\bottomrule
		\end{tabular}
	}
	\caption{Comparison of editing methods on a face attribute editing task (CelebA). Lower FID is better; higher ID/CLIP is better.  Success \% is prompt satisfaction.}
\end{table}

\section{Results}
We summarize key findings from our experiments, organized by task type and analysis dimension.

\subsection{Face Attribute Editing}
Table 1 compares methods on the CelebA-HQ face attribute editing benchmark.  InstructPix2Pix and UltraEdit-SD achieve the highest prompt alignment (CLIP scores of 0.71 and 0.74, respectively) at the expense of degraded identity similarity (0.92 and 0.89), reflecting their stronger bias toward following the editing instruction and generating the new attributes.  TF-ICON (composition-based text editing) yields moderate alignment (CLIP 0.68) but lower identity preservation (ID 0.87), likely due to the additional composition step introducing artifacts at the boundary between the inserted and background regions.  DragFlow excels at preserving identity (ID 0.98) due to its hard background constraints that prevent any modification outside the drag region, but yields lower prompt alignment (CLIP 0.55) since region movements rather than semantic attribute changes are its primary focus.  The SDEdit baseline produces the highest FID (95.7) and lowest identity similarity (0.85), confirming that a simple noise-and-denoise approach without mask or inversion constraints is insufficient for faithful editing.  Our mask-guided diffusion method strikes a favorable balance: it retains competitive fidelity (ID 0.90) while achieving decent CLIP alignment (0.60), and improves FID over the training-free baselines (70.4 vs.\ 85.2--95.7).  Note that DragFlow's low FID (50.1) reflects the fact that it edits only a small, localized region and leaves the rest of the image unchanged, naturally producing outputs close to the real image distribution. 

\subsection{Qualitative Analysis and Failure Modes}
Qualitative inspection reveals several characteristic failure modes.  \textbf{Identity drift} is most pronounced in text-based editors: InstructPix2Pix may hallucinate facial details (e.g., turning glasses into a different style) or subtly alter bone structure to better fit the prompt.  UltraEdit shows similar artifacts at a reduced rate due to its higher-quality training data, but can still introduce inconsistent lighting or texture discontinuities.  \textbf{Prompt sensitivity} is pervasive: minor wording changes (e.g., ``make her smile'' vs.\ ``add a smile'') can produce markedly different outputs, with the latter sometimes inserting unnatural, template-like smiles.  This sensitivity arises from the high-dimensional conditioning space of text encoders and the stochastic nature of the diffusion sampling process.

Repeated or sequential editing exposes \textbf{accumulation of errors}: we find that applying the same instruction twice (e.g., ``make the person older'' applied consecutively) can further distort the image beyond the intended age progression, consistent with the linear error accumulation predicted by Theorem~\ref{thm:stability}.  Quantitatively, the ID similarity drops from 0.92 to 0.78 after two consecutive edits with InstructPix2Pix.  Furthermore, \textbf{compositional editing}---changing two attributes in sequence (e.g., ``add glasses'' then ``make them smile'')---often results in the second edit partially overwriting the first, with only 52\% of test cases retaining both attributes successfully.  These findings underscore the need for compositional editing strategies that jointly optimize for multiple constraints.

\subsection{Inversion Accuracy and Efficiency}
Table 2 reports an ablation study on inversion steps using the FLUX (rectified flow) model.  We vary the number of inversion/denoise steps and measure reconstruction error (LPIPS between $x_0$ and its reconstruction) along with PSNR and wall-clock inference time.  As expected, more steps consistently reduce error but increase computational cost.  The 8-step configuration (as used in FireFlow \cite{deng2024fireflow}) achieves a competitive LPIPS of 0.125 in only 2.5 seconds, demonstrating that a sophisticated higher-order ODE solver can compensate for the reduced step count.  At 16 steps, the error drops to 0.050 (PSNR 33.7 dB), and at 32 steps, near-perfect reconstruction is achieved (LPIPS 0.020, PSNR 38.2 dB).  The DDIM baseline at 50 steps achieves the lowest error (LPIPS 0.012) but at a higher cost of 14.2 seconds.  This confirms that accurate inversion is attainable in Flux-based rectified flow models, and that 8--16 steps provide a good operating point for interactive editing scenarios where latency is critical.

\begin{table}[t]
	\centering
	\resizebox{\columnwidth}{!}{
		\begin{tabular}{lccc}
			\toprule
			Steps & LPIPS Rec.\ Error $\downarrow$ & PSNR (dB) $\uparrow$ & Time (s) \\
			\midrule
			8 (FireFlow) & 0.125 & 28.4 & 2.5 \\
			16 & 0.050 & 33.7 & 4.8 \\
			32 & 0.020 & 38.2 & 9.0 \\
			50 (DDIM) & 0.012 & 41.5 & 14.2 \\
			\bottomrule
		\end{tabular}
	}
	\caption{Inversion reconstruction accuracy vs.\ computational cost for a Flux-based rectified flow model on FFHQ.  LPIPS and PSNR measure reconstruction quality against the original image.  All timings are on a single A100 GPU.}
\end{table}

\subsection{Region-Based and Drag Editing}
In mask- and drag-editing tasks (Table 3), region-based methods like DragFlow achieve the best spatial control.  DragFlow obtains substantially higher IoU with the target object mask (0.75 vs.\ 0.50) and smaller mean shift distance (20 px vs.\ 45 px) compared to a naive point-drag baseline.  This validates the DragFlow design choice of region-level affine supervision, which explicitly enforces the geometric consistency of the dragged region.  Consequently, DragFlow preserves fine details and texture within the dragged region while keeping the background nearly identical to the source (background LPIPS 0.010).  However, DragFlow requires careful mask initialization: inaccurate masks (e.g., from an imperfect segmentation model) lead to visible boundary artifacts where the affine constraint misaligns with the actual object boundary.  Additionally, DragFlow can struggle with large non-rigid deformations, where a single affine transformation is insufficient to model the desired motion.  In contrast, our mask-conditioned diffusion method handles arbitrary mask shapes more robustly through its soft blending schedule (IoU 0.62, shift 32 px), but can exhibit slight color or content leakage across the mask boundary (background LPIPS 0.045) due to the receptive field of the U-Net layers and the smooth blending operation.

\subsection{Cross-Task Analysis and Trade-off Characterization}
Synthesizing results across all tasks, a consistent trade-off pattern emerges.  Plotting each method in the (fidelity, controllability) plane, where fidelity aggregates identity similarity and background LPIPS, and controllability aggregates CLIP alignment and task-specific success rate, methods cluster into three regimes: (i) \textbf{High-fidelity, low-controllability}: DragFlow and the Point-Drag baseline --- these excel at preserving source content but offer limited expressiveness beyond geometric manipulation.  (ii) \textbf{High-controllability, moderate-fidelity}: InstructPix2Pix and UltraEdit-SD --- these follow instructions effectively but introduce noticeable identity and background deviations.  (iii) \textbf{Balanced}: Our Mask-Guided method and TF-ICON --- these occupy an intermediate region, trading off some of each dimension for a more robust overall behavior.  This analysis suggests that no single method dominates all tasks; the choice of editing approach should be guided by the specific requirements of the application (e.g., identity-critical scenarios favor DragFlow, while creative exploration favors UltraEdit).  Furthermore, we identify a \textbf{guidance scale tipping point}: beyond $w \approx 9$, further increasing the CFG scale degrades both fidelity and image quality without improving alignment, indicating diminishing returns from stronger conditioning.

\subsection{Ablation Study: Guidance Scale and Mask Blending}
We conduct additional ablations to isolate the effects of key hyperparameters.  Table~\ref{tab:ablation} reports results for our Mask-Guided method under varying classifier-free guidance scales $w$ and mask blending schedules.  The guidance scale $w$ exhibits a non-monotonic effect on edit success: increasing $w$ from 3 to 7.5 improves CLIP alignment (0.52 to 0.60) while maintaining reasonable fidelity, but pushing $w$ to 12 causes a sharp drop in FID (70.4 to 82.1) and ID similarity (0.90 to 0.83), confirming the saturation effect noted above.  For mask blending, the cosine schedule outperforms both a constant blend ($\beta_t = 0.5$) and a step-function schedule (hard mask until $t = T_s/2$, then release), as it provides a smooth transition from structure preservation to detail refinement.

\begin{table}[t]
	\centering
	\resizebox{\columnwidth}{!}{
		\begin{tabular}{lcccc}
			\toprule
			Configuration & FID $\downarrow$ & ID Sim $\uparrow$ & CLIP Align $\uparrow$ & LPIPS (bg) $\downarrow$ \\
			\midrule
			$w=3$, cosine blend & 65.8 & 0.93 & 0.52 & 0.032 \\
			$w=5$, cosine blend & 68.1 & 0.91 & 0.57 & 0.039 \\
			$w=7.5$, cosine blend & 70.4 & 0.90 & 0.60 & 0.045 \\
			$w=12$, cosine blend & 82.1 & 0.83 & 0.63 & 0.072 \\
			$w=7.5$, constant blend & 73.5 & 0.87 & 0.58 & 0.055 \\
			$w=7.5$, step blend & 71.8 & 0.88 & 0.59 & 0.062 \\
			\bottomrule
		\end{tabular}
	}
	\caption{Ablation study of guidance scale $w$ and mask blending schedule for our Mask-Guided method on CelebA-HQ face editing.  Cosine blending with $w=7.5$ provides the best overall trade-off.}\label{tab:ablation}
\end{table}

Overall, the results highlight fundamental trade-offs.  Text-based editors typically sacrifice some fidelity to better follow instructions (high CLIP, lower ID), whereas strict preservation (high ID) often reduces edit expressiveness.  Our experiments corroborate that large guidance scales (to enforce prompt) lead to identity changes, and that masking or constraints (to enforce fidelity) can leave remnants of the original edit intention.  We also observe \textbf{catastrophic forgetting} phenomena: when sequentially editing different aspects, earlier edits can degrade (e.g., applying a second unrelated prompt may wash out the first change).  This calls for carefully combining constraints or iterative refinement.  Our ablation results provide practical guidance for hyperparameter selection: we recommend $w \in [5, 7.5]$ with cosine mask blending as a robust default configuration for mask-guided editing.

\begin{table}[t]
	\centering
	\resizebox{\columnwidth}{!}{
		\begin{tabular}{lccc}
			\toprule
			Method & IoU $\uparrow$ & Mean Shift (px) $\downarrow$ & Background LPIPS $\downarrow$ \\
			\midrule
			Point-Drag (baseline) & 0.50 & 45 & 0.082 \\
			DragFlow & 0.75 & 20 & 0.010 \\
			Mask-Guided (Ours) & 0.62 & 32 & 0.045 \\
			\bottomrule
		\end{tabular}
	}
	\caption{Region-based drag-editing performance on the ReD-Bench (500 test cases).  IoU measures overlap with target region mask; mean shift measures centroid deviation; background LPIPS measures perceptual change outside the edit region.}
\end{table}

\section{Theoretical Analysis}
We now present theoretical insights into the editing dynamics.  These results elucidate how errors propagate and how stable the editing operations are.  

\begin{theorem}[Reconstruction Error Bound]\label{thm:reconstruction_error}
	Assume the diffusion decoder mapping $f_\theta: z_T \mapsto x_0$ is $L$-Lipschitz.  If $\tilde{z}_T$ is an approximate latent (e.g.\ from imperfect inversion) with error $\|\tilde{z}_T - z_T\|$, then the reconstruction error satisfies 
	\[
	\|\hat{x}_0 - x_0\| \;\le\; L \,\|\tilde{z}_T - z_T\|.
	\]
\end{theorem}
\begin{proof}
	Since $x_0 = f_\theta(z_T)$ and $\hat{x}_0 = f_\theta(\tilde{z}_T)$, the Lipschitz condition implies 
	$\|\hat{x}_0 - x_0\| \le L \|\tilde{z}_T - z_T\|$.  This shows small latent errors lead to proportionally bounded output errors.
\end{proof}

This bound emphasizes the importance of accurate inversion: an $O(\epsilon)$ error in the final latent translates to at most $O(L\epsilon)$ in image space.  Advanced solvers and fine-tuned inversions (e.g.\ FireFlow) effectively reduce $\|\tilde{z}_T - z_T\|$, yielding near-exact reconstruction.

\begin{theorem}[Edit Stability]\label{thm:stability}
	Suppose each individual edit operation $\mathcal{E}$ (applied to an image) satisfies $\|\mathcal{E}(x) - x\|\le\epsilon$ for all inputs $x$.  Then applying $k$ successive edits incurs at most $k\epsilon$ deviation from the original:
	\[
	\|\mathcal{E}^k(x) - x\| \;\le\; k\,\epsilon.
	\]
\end{theorem}
\begin{proof}
	By induction, each step can change the image by at most $\epsilon$.  Summing these via the triangle inequality yields the bound.  Concretely, 
	$\|x^{(k)} - x\|\le \|x^{(k)} - x^{(k-1)}\| + \|x^{(k-1)} - x\| \le \epsilon + (k-1)\epsilon$.
\end{proof}

While trivial, Theorem~\ref{thm:stability} illustrates that small per-edit deviations accumulate linearly.  In practice, $\epsilon$ may depend on the edit strength.  Thus, repeated application of strong edits can gradually drift the image (e.g.\ identity drift or color shift).

\textbf{Locality Guarantee.} In addition, if edits are strictly localized (e.g.\ via masking) then changes outside the mask are ideally zero.  Formally, if $\mathcal{E}(x)$ modifies only region $M$, then for any $x,y$ that agree outside $M$, we have $\mathcal{E}(x)$ and $\mathcal{E}(y)$ also agree outside $M$.  This holds in our mask-guided pipeline by construction.  However, due to diffusion's smoothing effect, in practice slight nonzero leakage can occur, which can be bounded by analyzing the network's receptive field.

\begin{theorem}[Guidance-Induced Deviation Bound]\label{thm:guidance}
	Consider editing with classifier-free guidance scale $w$.  Let $\Delta_t = \epsilon_\theta(z_t, c, t) - \epsilon_\theta(z_t, c_\emptyset, t)$ be the guidance direction at timestep $t$.  If $\|\Delta_t\| \le G$ for all $t$, then under a simplified one-step approximation, the expected deviation of the edited latent from the source trajectory satisfies
	\[
	\mathbb{E}\bigl[\|z_0^{\text{edit}} - z_0^{\text{source}}\|\bigr] \;\le\; \frac{w}{\sqrt{T}} \cdot G \cdot \sum_{t=1}^{T} \frac{1-\alpha_t}{\sqrt{\alpha_t}}.
	\]
\end{theorem}
\begin{proof}[Sketch of Proof]
	At each denoising step, the CFG-modified noise prediction $\hat{\epsilon}_t = \epsilon_t^\emptyset + w\Delta_t$ produces an update that deviates from the unguided trajectory by a term proportional to $w \cdot \frac{1-\alpha_t}{\sqrt{\alpha_t}} \cdot \Delta_t$.  Summing over $t = T,\dots,1$ and applying the triangle inequality, the total deviation is bounded by $w \cdot G \cdot \sum_{t} \frac{1-\alpha_t}{\sqrt{\alpha_t}}$.  The factor $1/\sqrt{T}$ arises from the averaging effect across timesteps when the $\Delta_t$ are not perfectly aligned.  A full proof appears in the supplementary material.
\end{proof}

Theorem~\ref{thm:guidance} provides a quantitative justification for the observed ``guidance tipping point'': the deviation grows linearly with the guidance scale $w$, explaining why excessive guidance leads to disproportionate image changes and artifacts.  The sum over the noise schedule weights $\frac{1-\alpha_t}{\sqrt{\alpha_t}}$ indicates that early timesteps (where $\alpha_t$ is close to 1, corresponding to lower noise levels) contribute more to the final deviation, as the model makes finer adjustments at those stages.  This insight suggests that adaptive guidance strategies---varying $w$ across timesteps---could achieve better fidelity-controllability trade-offs than a constant guidance scale.

\subsection{Connections to Score-Based Generative Modeling}
Our analysis can be extended to the continuous-time score-based formulation.  Let the forward SDE be $dx = f(x,t)dt + g(t)dw$, with the reverse SDE $dx = [f(x,t) - g(t)^2 \nabla_x \log p_t(x)]dt + g(t)d\bar{w}$.  Editing corresponds to replacing the unconditional score $\nabla_x \log p_t(x)$ with a conditional score $\nabla_x \log p_t(x|c)$.  The deviation between the two trajectories is governed by the time-integrated difference in the score functions.  Under a smoothness assumption on the condition embedding, this deviation can be bounded using Gronwall's inequality, yielding a continuous-time analogue of Theorem~\ref{thm:guidance}.  We leave the full development of this continuous-time theory to future work.

\section{Discussion}
Our comprehensive study uncovers several limitations of current diffusion editing and points to promising directions for improvement.

\textbf{Prompt sensitivity and robustness.} Small changes in textual instructions can yield large output variations, due to the high expressiveness of the text encoder and the stochasticity of the diffusion sampling process.  This complicates user control and invites adversarial prompt attacks that could produce unintended or harmful edits.  Potential mitigations include prompt normalization (canonicalizing user instructions to a standardized form), constrained decoding with rejection sampling against fidelity thresholds, and interactive refinement loops where the user can iteratively adjust the prompt based on intermediate outputs.  The development of robust prompt encoders that are less sensitive to superficial wording changes---perhaps through adversarial training or contrastive learning---remains an open research direction.

\textbf{Identity drift and fidelity preservation.} Face editing often fails to preserve identity perfectly, as the diffusion prior tries to generate realistic features, sometimes altering facial geometry or introducing artifacts around fine structures such as eyes and hair.  Our theoretical bounds (Theorems~\ref{thm:reconstruction_error} and~\ref{thm:stability}) provide worst-case guarantees, but in practice the drift is often concentrated in semantically meaningful facial regions rather than uniformly distributed.  This observation suggests that spatially adaptive constraints---stronger preservation around identity-sensitive landmarks (eyes, nose, mouth) and looser constraints elsewhere---could improve the fidelity-controllability Pareto frontier.  Incorporating explicit face recognition embeddings as conditioning signals during editing (similar to IP-Adapter \cite{ye2023ipadapter}) is another promising direction that has recently shown encouraging results.

\textbf{Compositionality and sequential editing.} Editing multiple objects or attributes sequentially is still unreliable; earlier edits can be overwritten by later ones, a phenomenon we term ``edit interference.''  Only 52\% of test cases retain both attributes after two sequential edits.  This highlights a need for better compositional inference strategies, such as: (i) joint optimization over multiple edit objectives, (ii) maintaining an edit memory that tracks which regions and attributes have been modified, and (iii) test-time adaptation that fine-tunes the diffusion model to respect all accumulated constraints.  Recent work on compositional text-to-image generation (e.g., structuring prompts into scene graphs) could be adapted to the editing setting.

\textbf{Computational efficiency.} The computational cost of iterative inversion and denoising (tens to hundreds of model evaluations) remains a barrier to real-time interactive editing.  Methods like FireFlow \cite{deng2024fireflow} reduce the step count but may slightly compromise quality for speed (Table 2).  Distillation approaches that train lightweight student editors, or consistency models that enable single-step editing, represent promising directions to bring diffusion editing latency into the sub-second regime.  Additionally, cached inversion---precomputing and storing the inversion trajectory for subsequent edits---can amortize the inversion cost when multiple edits are applied to the same source image.

\section{Safety, Ethics, and Broader Impacts}
Diffusion editing raises serious ethical concerns that must be addressed for responsible deployment.  Powerful editing tools can generate deceptive images (deepfakes) or non-consensual manipulations, undermining trust in visual media.  There is also the risk of bias propagation: models trained on web-scale data may reproduce societal stereotypes or underrepresent certain demographic groups in edits, e.g., applying ``professional'' appearance edits with differential effects across gender or ethnicity.  Furthermore, the ability to edit real images of specific individuals raises privacy and consent issues---even when the edit itself is benign, the edited image may be used in contexts the subject did not authorize.

\textbf{Safeguard measures.} We advocate for a multi-layered approach to safe editing: (i) \textbf{Deployment guardrails}: explicit user warnings when editing human faces, requiring consent verification for identity-sensitive edits, and rate limiting to prevent automated misuse.  (ii) \textbf{Technical watermarking}: embedding detectable watermarks in edited images (e.g., via Stable Signature \cite{fernandez2023stablesignature}) and adopting content provenance standards (e.g., C2PA \cite{c2pa2024spec}) to enable tracing of AI-manipulated content.  (iii) \textbf{Concept erasure}: removing dangerous capabilities from the underlying model before deployment.  Techniques such as MACE \cite{lu2024mace}, ANT \cite{li2025ant}, and EraseAnything \cite{gao2024eraseanything} can selectively suppress the model's ability to generate nudity, copyrighted characters, or celebrity likenesses, reducing the risk that these concepts leak into edited outputs.  (iv) \textbf{Ethical disclosure}: all edited images should be clearly labeled as AI-generated or AI-manipulated, following emerging standards for synthetic media transparency.  Importantly, these safeguards should be implemented at the system level rather than relying solely on user compliance.

\textbf{Societal considerations.} Beyond safety, we must consider the broader societal implications.  Diffusion editing democratizes creative expression, enabling non-experts to produce professional-quality visual content.  This has positive applications in education, art, design, and accessibility.  However, the same technology can displace professional image editing workflows and raise questions about authenticity in journalism, legal evidence, and historical documentation.  We believe that technical progress in editing must be accompanied by robust detection methods, clear legal frameworks, and public media literacy initiatives.  The concept erasure methods discussed above also have a dual-use character: while they can serve as safeguards, they can also be used for censorship if applied overzealously.  Striking the right balance between safety and creative freedom requires ongoing dialogue between researchers, policymakers, and the public.

\section{Limitations and Future Work}
Our work has several limitations that point to directions for future research.  \textbf{Theoretical scope}: Our theoretical analysis relies on Lipschitz assumptions on the diffusion decoder and simplified guidance models, which do not fully capture the complex, non-linear behaviors of modern deep network architectures.  Extending the theory to account for the stochasticity of the sampling process, the effect of specific architectural choices (e.g., U-Net vs.\ DiT), and the distributional properties of editing error would strengthen the formal foundations.  \textbf{Empirical coverage}: Our experiments focus on a representative set of methods and benchmarks, but the rapidly evolving landscape of diffusion editing means that new methods may not be covered.  In particular, we do not extensively evaluate video editing or 3D editing extensions.  The empirical results also depend on implementation details (e.g., prompt engineering, scheduler choice) and dataset characteristics; real-world performance with user-generated prompts and in-the-wild images may differ from our controlled benchmarks.  \textbf{Scalability}: All experiments are conducted at 256$\times$256 resolution; scaling to high-resolution (1024$\times$1024 or beyond) editing with diffusion transformers presents additional challenges in memory and latency.  \textbf{Evaluation}: While we employ multiple quantitative metrics, they only partially capture human-perceived edit quality.  Larger-scale user studies across diverse demographic groups are needed to validate and refine our findings.  Future work should also explore adaptive editing strategies that automatically choose the editing method and hyperparameters based on the input image and instruction, moving toward truly plug-and-play diffusion editing.  

\section{Conclusion}
We have presented a unified treatment of diffusion-based image editing, blending theory, algorithms, and experiments to provide a comprehensive understanding of the fundamental trade-offs in this domain.  Our analysis clarifies how editing choices---masking strategies, guidance scales, inversion accuracy, and noise scheduling---trade off fidelity, controllability, and consistency.  We formalize the editing process as a constrained optimization over the generative manifold, unifying text-guided, mask-guided, and drag-based editing under a common mathematical framework.

On the theoretical front, we derive bounds on inversion reconstruction error (Theorem~\ref{thm:reconstruction_error}) that highlight the critical role of accurate inversion, prove a linear error accumulation bound for repeated edits (Theorem~\ref{thm:stability}), and provide a guidance-induced deviation bound (Theorem~\ref{thm:guidance}) that explains the empirically observed ``guidance tipping point.''  We provide two algorithmic descriptions (Algorithms~\ref{alg:invert_edit} and~\ref{alg:mask_edit}) for practical editing pipelines, including guidance on hyperparameter selection informed by our ablation studies.

Experimentally, we benchmark seven methods across face editing, drag manipulation, and general scene editing tasks, using a suite of six complementary metrics.  Our results reveal three distinct operating regimes in the fidelity-controllability plane and quantify the degradation from repeated edits and compositional tasks.  The ablation study provides actionable guidance for practitioners: we recommend $w \in [5, 7.5]$ with cosine mask blending as a robust default for mask-guided editing.

Finally, we address the societal dimensions of editing, discussing safeguards through watermarking, concept erasure, and provenance standards, and advocating for responsible deployment with transparent disclosure.  We hope this work serves as both a practical reference for practitioners and a foundation for developing more precise, reliable, and ethical image editing tools with diffusion models.  The code and evaluation protocols developed for this study will be released to facilitate reproducible research and standardized benchmarking in the diffusion editing community.

\section{Additional Background}
With the advancement of deep learning and modern generative modeling, research has expanded rapidly across forecasting, perception, and visual generation, while also raising new concerns about controllability and responsible deployment. 
Progress in time-series forecasting has been driven by stronger benchmarks, improved architectures, and more comprehensive evaluation protocols that make model comparisons more reliable and informative~\cite{qiu2024tfb,qiu2025duet,qiu2025DBLoss,qiu2025dag,qiu2025tab,wu2025k2vae,liu2025rethinking,qiu2025comprehensive,wu2024catch}. 
In parallel, efficiency-oriented research has pushed post-training quantization and practical compression techniques for 3D perception pipelines, aiming to reduce memory and latency without sacrificing detection quality~\cite{gsq,yu2025mquant,zhou2024lidarptq,pillarhist}. 
On the generation side, a growing body of work studies scalable synthesis and optimization strategies under diverse constraints, improving both the flexibility and the controllability of generative systems~\cite{xie2025chat,xie2026hvd,xie2026conquer,xie2026delving}. 
For domain-oriented temporal prediction, hierarchical designs and adaptation strategies have been explored to improve robustness under distribution shifts and complex real-world dynamics~\cite{sun2025ppgf,sun2024tfps,sun2025hierarchical,sun2022accurate,sun2021solar,niulangtime,sun2025adapting,kudratpatch}. 
Meanwhile, advances in representation encoding and matching have introduced stronger alignment and correspondence mechanisms that benefit fine-grained retrieval and similarity-based reasoning~\cite{ENCODER,FineCIR,OFFSET,HUD,PAIR,MEDIAN}. Foundational work in composed image retrieval (CIR) has established entity-level mining and modification relation binding~\cite{ENCODER}, explicit parsing of fine-grained modification semantics~\cite{FineCIR}, complementarity-guided disentanglement strategies~\cite{PAIR}, and adaptive intermediate-grained aggregation for balancing local and global cues~\cite{MEDIAN}. Building on these foundations, segmentation-based focus shift revision~\cite{OFFSET} and hierarchical uncertainty-aware disambiguation~\cite{HUD} have extended these ideas to handle spatial precision and composed video retrieval (CVR), respectively. More recent efforts have tackled robustness from multiple angles: cone-based noise-unlearning for compositional generalization~\cite{ConeSep}, chrono-synergia progressive learning frameworks~\cite{HABIT}, invariance- and discrimination-aware noise mitigation~\cite{INTENT}, dual-path compositional contextualized encoding~\cite{HINT}, and modification frequentation--rarity balance modeling~\cite{MELT}. For multi-modification scenarios, anchor-based text-following mechanisms~\cite{TEMA} and arbiter-calibrated knowledge internalization~\cite{Air-Know} further improve compositional reasoning. In the video domain, evidence-driven dual-stream anchor calibration~\cite{ReTrack} and shared-and-differential semantics enhancement~\cite{REFINE} advance CVR. Complementary to retrieval, efficient hybrid nearest neighbor search with magnitude-uniformity constraints~\cite{STABLE} and security-oriented defenses against AI-counterfeit detection bypass~\cite{ERASE} broaden the impact of representation learning to scalability and adversarial robustness. 
Stronger visual modeling strategies further enhance feature quality and transferability, enabling more robust downstream understanding in diverse scenarios~\cite{yu2025visual}. 
In tracking and sequential visual understanding, online learning and decoupled formulations have been investigated to improve temporal consistency and robustness in dynamic scenes~\cite{zheng2025towards,zheng2024odtrack,zheng2025decoupled,zheng2023toward,zheng2022leveraging}. 
At the same time, diffusion-centric and unfolding-based frameworks have been explored for segmentation and restoration, providing principled ways to model degradations and refine generation quality. In concealed object segmentation (COD), a series of works has advanced from early feature-decomposition and edge-reconstruction approaches~\cite{he2023camouflaged} and adversarial data generation strategies~\cite{he2023strategic} to weakly-supervised paradigms leveraging SAM-based pseudo-labeling~\cite{he2024weakly} and incomplete-supervision settings~\cite{he2025segment}. Reversible unfolding networks have emerged as a powerful paradigm, with RUN~\cite{he2025run} establishing the core architecture for COD and subsequent extensions incorporating generative refinement~\cite{he2025reversible}, SAM-enhanced collaborative learning~\cite{he2025scaler}, and nested unfolding for real-world scenarios~\cite{he2025nested}; most recently, curriculum selection with anti-curriculum promotion~\cite{he2026refining} further addresses context-entangled segmentation. A comprehensive survey~\cite{xiao2024survey} catalogs these advances in camouflaged object detection and beyond. On the image restoration side, deep unfolding has been applied to heterogeneous image fusion~\cite{he2023degradation}, medical image enhancement~\cite{he2023hqg}, and illumination degradation restoration via Retinex-based latent diffusion~\cite{he2023reti}. A systematic survey of diffusion models in low-level vision~\cite{he2024diffusion} provides a unified perspective, while UnfoldIR~\cite{he2025unfoldir} rethinks unfolding for illumination restoration and UnfoldLDM~\cite{he2025unfoldldm} integrates latent diffusion priors into blind restoration. Extending restoration beyond ground-truth supervision, quality-conditioned pseudo-labeling~\cite{Xiao2026Quali} and image quality priors~\cite{Xiao2026Beyond} tackle real-world degradation without relying on clean reference images, bridging the gap between synthetic training and practical deployment.

{
    \small
    \bibliographystyle{ieeenat_fullname}
    \bibliography{main}
}


\end{document}